\documentclass[letterpaper]{article} 
\usepackage{aaai25}
\usepackage{times}  
\usepackage{helvet}  
\usepackage{courier}  
\usepackage[hyphens]{url}  
\usepackage{graphicx} 
\usepackage{natbib}  
\usepackage{caption} 
\usepackage{soul}
\usepackage{url}
\usepackage[utf8]{inputenc}

\usepackage{amsmath}
\usepackage{amsthm}
\usepackage{booktabs}
\usepackage[switch]{lineno}

\usepackage{thmtools}

\usepackage{thm-restate}

\usepackage[utf8]{inputenc}
\usepackage{amsthm}
\usepackage{amsmath}
\usepackage{amssymb}
\usepackage{graphicx}
\usepackage{comment}

\usepackage[linesnumbered,ruled,vlined]{algorithm2e}
\SetAlgoNlRelativeSize{-1}

\newtheorem{theorem}{Theorem}[section]
\newtheorem{definition}{Definition}[section]
\newtheorem{claim}{Claim}[section]
\newtheorem{corollary}{Corollary}[section]
\newtheorem{lemma}{Lemma}[section]
\newtheorem{example}{Example}[section]

\title{Beyond Proportional Individual Guarantees for Binary Perpetual Voting}

\author{
Yotam Gafni\textsuperscript{\rm 1}, Ben Golan\textsuperscript{\rm 2}
}

\affiliations{
    \textsuperscript{\rm 1}Weizmann Institute\\
        \textsuperscript{\rm 2}Leo Baeck High-school\\

}

\usepackage{bibentry}
\nocopyright

\begin{document}

\maketitle

\begin{abstract}
   Perpetual voting studies fair collective decision-making in settings where many decisions are to be made, and is a natural framework for settings such as parliaments and the running of blockchain Decentralized Autonomous Organizations (DAOs). We focus our attention on the binary case (YES/NO decisions) and \textit{individual} guarantees for each of the participating agents. 
We introduce a novel notion, inspired by the popular maxi-min-share (MMS) for fair allocation. The agent expects to get as many decisions as if they were to optimally partition the decisions among the agents, with an adversary deciding which of the agents decides on what bundle. We show an online algorithm that guarantees the MMS notion for $n=3$ agents, an offline algorithm for $n=4$ agents, and show that no online algorithm can guarantee the $MMS^{adapt}$ for $n\geq 7$ agents. We also show that the Maximum Nash Welfare (MNW) outcome can only guarantee $O(\frac{1}{n})$ of the MMS notion in the worst case.
\end{abstract}

%

\section{Introduction}

A main focus of social choice theory is fair voting. Mostly, fair voting focuses on how to accurately represent the public's preferences when choosing a single candidate or a committee (see, e.g., Part I of the Handbook of Computational Social Choice \cite{handbookCSC}). However, once this voting body is established, different representatives or parties may have varying opinions on a large set of issues. It is generally accepted that the majority (or, in some cases, a qualified majority) vote within the voting body should decide the fate of each decision. Seemingly innocuous, this approach may lead to the following conundrum: Imagine that within a voting body of $100$ representatives, virtually all decisions are supported by the same $51$ of the representatives, while $49$ object. In this case, by the end of the session of Parliament, the thousands of decisions will all go in one way, even though the majority's advantage is minuscule. Compare this to if we implemented a random dictator process. Then, we could expect that a little more than half of the decisions would go one way, and little less than half the other way. Our main question in this work is: Could we devise appropriate fairness notions for such decision-making situations? And can we find deterministic, algorithmic ways of implementing them? 

This problem is especially timely due to the growing interest in advanced governance mechanisms for Blockchain Decentralized Autonomous Organizations (DAOs) \cite{fabrega2023DAO}. DAOs are smart contracts with assets under management, whose use is decided by voting tokens issued to stakeholders. 
As of the time of writing, DAOS are responsible for managing over 20 Billion USD in assets \cite{DeepDAO}. Typically, DAOs have a few dozens/hundreds of voters, voting (over time) over a few thousands of proposals (see Table~1 of \cite{grennan2023DAO}). Respecting minority opinion is crucial for DAOs, as it could otherwise result in ``rage quitting'' \cite{fabrega2023DAO}, as already happened, e.g., in NounsDao. Rage quitting has many adverse affects, one of which is the consolidation of power by the majority token holders. Other solutions include veto rights for the minority: But this has other disadvantages, as decisions can be stalled indefinitely. It is thus recognized that better governance mechanisms are needed. 

\subsection{Our Approach}

We wish to focus on the most simple additive model. As we see, this still leads to interesting and non-trivial results and techniques. In particular, we impose three important restrictions over the possible decisions: 

\begin{enumerate}
\item Each decision is binary, i.e., has two possible choices, 

\item Each decision implies polarized preferences. I.e., if the agent's preference is satisfied, it receives a utility of $1$, and $0$ otherwise, 

\item Each decision has equal additive weight. 
\end{enumerate} 

We then consider our main notion, inspired by maxi-min-share (MMS) for indivisible goods:

The agent considers an adversary that assigns different sets of decisions to different agents. The agent chooses a partition of the decisions to $n$ different bundles, and the adversary chooses which agent decides which bundle. 

\begin{example}[Informal]
\label{ex:mms_motivation}
If all agents \textit{agree} over all $m$ decisions, an agent can expect all decisions to go her way. Indeed, under any partition and any agent assignment by the adversary, all decisions will follow the consensus and will go according to the agent's preference.

If the agent is \textit{opposed} to all other $n-1$ agents on all decisions, the agent can expect $\lfloor \frac{m}{n} \rfloor$ of the decisions to go her way. This is achieved by partitioning the decisions roughly equally, whereas the worst adversary choice is to let the agent decide a set of size rounded down to $\lfloor \frac{m}{n} \rfloor$. 
\end{example}

As we see in the last example, our MMS notion guarantees the agent at least PROP1 (their proportional share of the decisions, up to one decision), which is why we say it goes \textit{beyond the proportional} individual guarantee.

\subsection{Our Results}
In Section~\ref{sec:prelim}, we define our model and share notion. In Section~\ref{sec:existence}, we develop existence results and algorithms. 
We focus on our main class of algorithms, which we call \textit{graceful sequences}. In a graceful sequence algorithm, we define for each decision \textit{type} (i.e., a combined preference profile of all agents), a cyclical sequence of rules on how to decide it. For example, with $n=3$, we consider a sequence where, if two agents support a motion and the other agent opposes it, then the motion passes the first two times, but is rejected the third time, and so on. For $n=3$, this (online) algorithm indeed achieves MMS. For $n=4$, no online algorithm achieves MMS, but an \textit{offline} algorithm that uses it as a primitive achieves MMS. For $n\geq 7$, we show that no online algorithm can achieve the MMS. 

In Section~\ref{sec:mnw}, we study the comparison between our MMS notion and the individual guarantees of the Max Nash welfare (MNW) outcome. In many important examples the two notions yield the same outcome, for example, in Example~\ref{ex:mms_motivation}. However, we give examples where their behavior differs, and 
prove that the Max Nash welfare allocation may only guarantee $O(\frac{1}{n})$ approximation.

\subsection{Related Work}

\textit{Fair Public Decision-Making.} \cite{fairPublicDM} present a general model for fair decision-making that contains our model as a special case, as it considers a general number of options, and general positive rewards. Within this more specialized setting, their share notion hierarchy (Round-Robin Share $\implies$ Pessimistic Proportional Share $\implies$ PROP1) collapses to PROP1. Notably, while they do not show results for it, they define a notion of MMS similar to ours, with one important difference: While they assume the agents get utility $0$ for any decision not assigned to them by the adversary, we give the adversary a more limited power: The adversary may assign the other decision bundles to the other agents, but the agents decide over them. This way, an agent can expect more decisions to go their way if they're in agreement with the other agents. In their results, MNW guarantees PROP1, but may only guarantee $\frac{1}{n}$ of the Pessimistic Proportional Share. Since they are the same in our settings, MNW guarantees both, but still only guarantees $O(\frac{1}{n})$ of our MMS notion. 


\textit{Perpetual Voting.} \cite{lackner2020PV,LacknerMaly2023,leximinPV} suggest another model of which we are a special case. In this model, with each decision, voters approve one of $n$ options, and get binary utilities depending on whether one of their approved options was accepted for that decision. We only allow $n=2$ options, and do not allow voters to approve both options. In terms of the type of questions asked, \cite{lackner2020PV} focuses on an \textit{axiomatic} analysis: They define some natural voting rules, and check whether they respect different axioms, such as always deciding with the consensus if such exists, or deciding proportionally if the same decision-type is repeated $n$ times. They also introduce a share-notion they call ``Perpetual Lower Quota'', but show that no approximation of it can be guaranteed. As we show, this is a \textit{strictly} stronger notion of our MMS notion. \cite{FritschPV,SkowronGorecki2022} focus on the binary setting of perpetual voting, as we do.  

\textit{Group Proportionality in Perpetual Voting.} \cite{Bulteau2021JR,Peters2024EJR,FKP,masarik2023generalisedEJR,SkowronGorecki2022} study notions of proportionality for \textit{groups}, in the spirit of \textit{Justified Representation} \cite{aziz2017EJR}, and appropriate algorithms to guarantee them, such as the \textit{Method of Equal Shares}. We demonstrate the difference of our approach through an example.

\begin{example}
\label{ex:jr_vs_mms}
Consider $n=3$ agents and $m=9$ decisions:
$$\begin{bmatrix} 1 & 1 & 0 & 1 & 1 & 0 & 1 & 1 & 0 \\
1 & 1 & 1 & 1 & 1 & 1 & 1 & 1 & 1 \\ 0 & 0 & 1 & 0 & 0 & 1 & 0 & 0 & 1\end{bmatrix}.$$
(i.e., agents $1$ and $2$ support the first decision, while agent $3$ objects to it, and so on).

By Justified Representation, agent $1$ and $2$ are $\frac{2}{3}$ of the agents, and they agree on $\frac{2}{3}$ of the decisions, and thus should be regarded as a cohesive block. Thus, at least one of them (if we consider EJR - Extended Justified Representation \cite{aziz2017EJR}) should have $6$ decisions go their way, while the other does not have an individual guarantee. Agent $3$ does not form a bloc with the other agents, and so is guaranteed $\frac{1}{3}$ of the decisions. 

With our MMS notion, the picture is different. Agent $3$ is guaranteed $4$ decisions to go their way, by bundling decisions $1-3, 4-6,$ and $7-9$ together. Similarly, $MMS_1 = 5, MMS_2 = 6$. Thus, agent $3$ gains (in comparison to EJR) since it does agree with agent $2$ on some decisions, even though they do not form a bloc. Agent $2$ gets the EJR guarantee of $6$, and agent $1$ gets only $5$ since they are in the minority of some decisions (but this is still a better \textit{individual} guarantee than what EJR promises). 
\end{example}

\section{Preliminaries}
\label{sec:prelim}

Missing proofs appear in the appendix.

We wish to focus on the most simple additive model. In particular, we impose three important restrictions over the possible decisions: 

\begin{enumerate}
\item Each decision is binary, i.e., has two possible choices, 

\item Each decision implies a polarized preferences, i.e., if the agent's preference is satisfied, it receives a utility of $1$, and $0$ otherwise, 

\item Each decision has an equal additive weight. 

\end{enumerate} 

Formally:
\begin{definition}
We consider $[m] = \{1,\ldots, m\}$ decisions and $[n]$ agents. 

For any decision $j\in [m]$ there are two choices, $0$ and $1$.

The agents preference matrix $M \in \{0,1\}^{n \times m}$ is a $0-1$ matrix with $n$ rows and $m$ columns, so that row $M_i \in \{0,1\}^m$ is agent $i$'s preference vector over all decisions. 

An outcome $A \in \{0,1\}^m$ of a decision-making problem is a binary $m$-tuple that determines each of the $m$ decisions.

The Hamming distance between two outcomes $A,B$ restricted to the indices $X$ is: $d_H(A,B ; X) = \sum_{k\in X} 1[A[k] = B[k]]$. If we do not specify $X$, then we assume $X = [m]$. 

The utility of agent $i$ from outcome $A$ is $$u_i(A) = d_H(M_i, A),$$ the $\#$ of decisions where the agent agrees with the outcome. 
\end{definition}

\subsection{An MMS notion for Decision-Making}


\begin{definition}
($MMS^{adapt}$)

Let $X_1, \ldots, X_n$ be a partition of $[m]$ (so that $\cup_{i=1}^n X_i = [m]$ and for any $i, j$, $X_i \cap X_j = \emptyset$ ). Let $S_n$ be all the permutations over $n$ bundles. 

$$
 MMS^{adapt}_i = \max_{\{X_1, \ldots, X_n\}} \min_{\sigma \in S_n} \sum_{j=1}^n d_H(M_i, M_{\sigma(j)} ; X_j).
$$

In words, adaptive MMS is the value of the game between agent $i$ that chooses a partition of the decisions, and an adversary that chooses the worst assignment of which agent gets to decide on which decision bundle. For each agent that $\sigma(j)$ that is assigned to a bundle $X_j$, agent $i$ ``gets'' the decisions in $X_j$ that they agree with $\sigma(j)$ about. 

We say that an outcome A is $\alpha$-MMS$^{adapt}$, if for any agent $i$, $u_i(A) \geq \alpha \cdot MMS^{adapt}_i$.
If an outcome is $\alpha$-MMS$^{adapt}$ with $\alpha = 1$, we simply say that it is MMS$^{adapt}$. 

\end{definition}


We compare the MMS to a share notion that we call the ``Random Dictator Share'':

\begin{definition}
\textit{Random Dictator Share} \[
\begin{split}
& RDS_i = E_{i'\sim UNI([n])}[u_i(M_{i'})] = \\
& \sum_{j=1}^m E_{i'\sim UNI([n])} 1[M_{i'}[j] = M_i[j]] = \\
& \sum_{j=1}^m \frac{1}{n}\sum_{i'=1}^n 1[M_{i'}[j] = M_i[j]] ,
\end{split}
\]

i.e., it is the expected utility of agent $i$ if we let a dictator chosen uniformly at random decide everything.  
\end{definition}

The random dictator share is the same as what \cite{lackner2020PV,LacknerMaly2023} define as upper and lower quota. What we show next is that $MMS^{adapt}$ is ``upper quota compliant'', or, in our terms:

\begin{lemma}
For any $n$ agents, $m$ decisions, and preference matrix $M$, it holds that for any agent $i$:
$MMS^{adapt}_i \leq RDS_i$.
\end{lemma}

\begin{proof}

The lemma is a corollary of the following general claim:

\begin{claim} 
\label{clm:expectation_bound}
For any distribution $D$ over permutations, 
$$MMS_i^{adapt} \leq \max_{\{X_1, \ldots, X_n\}} \lfloor E_{\sigma \sim D}[\sum_{j=1}^n d_H(M_i, M_{\sigma(j)} ; X_j) \rfloor $$
\end{claim}
\begin{proof}
The claim is immediate by the fact that for any fixed choice of partition $X_1, \ldots, X_n$, the worst adversary choice $\sigma$ is worse for the agent $i$ than any randomized adversary choice $D$. 
\end{proof}

Then, choosing $D = UNI(S_n)$, we have:

\[
\begin{split}
& MMS_i^{adapt} \leq \max_{\{X_1, \ldots, X_n\}} \lfloor E_{\sigma \sim UNI(S_n)}[\sum_{j=1}^n d_H(M_i, M_{\sigma(j)} ; X_j) ] \rfloor \\
& \leq \lfloor \sum_{j=1}^m E_{i'\sim UNI([n])} 1[M_{i'}[j] = M_i[j]] \rfloor = \lfloor RDS_i \rfloor \leq RDS_i.
\end{split}
\]
\end{proof}

The MMS value can be strictly lower than the RDS. We show this by an example. 
Proposition~4 of \cite{lackner2020PV} shows a generic example of why RDS can not always be guaranteed. We consider their example with $m=2$:

\begin{example}
\label{ex:mms_vs_rds}
Consider $n=4$ agents and $m=2$ decisions:
$$\begin{bmatrix} 1 & 1 \\
1 & 0  \\ 0 & 1 \\
0 & 0 \end{bmatrix}.$$

The RDS value of every agent is $1$ (as they have probability $\frac{1}{2}$ of getting each of the two decisions), while the MMS value of every agent is $0$. 
\end{example}

\section{Existence and Approximation of Adaptive MMS in Offline and Online Settings}
\label{sec:existence}

\begin{definition}
We say that an algorithm is \textit{online} if its decision at step $j$ only depends on the first $j$ columns of $M$. 
\end{definition}


\subsection{Small Values of $n$}

It is straight-forward to show that with $n=2$ agents, $MMS^{adapt}$ can be guaranteed, by deciding all consensus decisions with the consensus, and splitting the contested decisions between the two agent equally (up to a rounding error). We thus start by considering $n=3$.

We introduce the following notation: Let $\delta_i$ be the number of decisions $j$ so that $M_i[j] \neq M_{i+1 \mod 3}[j] = M_{i+2 \mod 3}[j]$, i.e., the number of decisions where agent $i$ opposes the other agents. Let $\delta = \delta_1 + \delta_2 + \delta_3$. 

\begin{corollary}
\label{corr:mms_adapt_n=3_bound}
For $n=3$, and every $1\leq i\leq 3$,
$$MMS_i^{adapt} \leq m - \delta +  \lfloor \frac{2}{3}\delta_2 + \frac{2}{3} \delta_3 + \frac{1}{3} \delta_1 \rfloor \leq m - \lceil \frac{\delta}{3} \rceil.$$

\end{corollary}

\begin{proof}
We consider w.l.o.g. the perspective of agent $1$. 
By Claim~\ref{clm:expectation_bound}, using the uniform distribution $D = UNI$, each decision has a probability of $\frac{1}{3}$ to end up being decided by a specific agent $j$. Thus, decisions of type $\delta_1$ have probability of $\frac{1}{3}$ to be decided in favor of agent $1$, while decisions of type $\delta_2$ and $\delta_3$ have a probability of $\frac{2}{3}$. Decisions that are unanimous are always decided in agent $1$'s favor and so they must be an integer. We get:

$MMS_1^{adapt} \leq \lfloor \frac{2}{3}\delta_2 + \frac{2}{3} \delta_3 + \frac{1}{3} \delta_1 + m - \delta \rfloor \leq \lfloor \frac{2}{3}\delta + m - \delta \rfloor \leq m - \delta + \lfloor \frac{2\delta}{3} \rfloor = m - \lceil \frac{\delta}{3} \rceil $. 

\end{proof}

\begin{definition}
\label{def:ptrra}
Per-Type Round-Robin Algorithm:

For the case of $n=3$, we define the following algorithm. When facing a unanimous decision, it goes by the consensus. For each of the $3$ types of non-unanimous decisions $\delta_i$ it keeps a dedicated counter $c_i$, and decides with the majority whenever the counter has $c_i \mod 3 \in \{0,1\}$, and with the minority when $c_i \mod 3 = 2$. 
\end{definition}

\begin{lemma}
\label{lem:mms_adapt_n=3}
For $n=3$, the per-type round-robin algorithm guarantees $MMS_i^{adapt}$. 
\end{lemma}

\begin{proof}
W.l.o.g., we consider the perspective of agent $1$. 
The algorithm guarantees 
\[
\begin{split}
    & u_1(ALG) = m - \delta + \lceil \frac{2}{3} \delta_2 \rceil + \lceil \frac{2}{3} \delta_3 \rceil + \lfloor \frac{1}{3} \delta_1 \rfloor \geq \\
    & m - \delta + \lfloor \frac{2}{3} \delta_2 + \frac{2}{3} \delta_3 + \frac{1}{3} \delta_1 \rfloor \stackrel{Corollary~\ref{corr:mms_adapt_n=3_bound}}{\geq} MMS_1^{adapt}.
\end{split}
\]

\end{proof}
 
The success of this simple algorithm in the case of $n=3$ may give hope that this could be generalized to larger values of $n$. We next show that this is not true, for a large class of algorithms that we call ``graceful sequences''. Intuitively, graceful sequences define a consistent way of determining, for each decision type, how will it be treated upon the first occurrence, second occurrence, and so on. For example, with $n=10$ agents, and for a decision type that has the first six agents prefer an outcome of $0$, and the others prefer an outcome of $1$, a graceful sequence may determine that the first six occurrences of this decision type will favor the majority, the next four will favor the minority, and so on. A different graceful sequence, for example, would alternate between majority and minority decisions.

\begin{definition}
Let $\rho: \{0,1\}^n  \rightarrow \{0,1\}^n$ be an even function (i.e., two boolean vectors that are a negation of each other are mapped in the same way). 

Then, the algorithm $A_{\rho}$ is a \textit{graceful sequence} based on $\rho$ if it cyclically maps the $k$ appearance of a decision-type $t$ (where in $k$ we account also for appearances of $t$'s negation) to $[\rho(t)]_{k \mod n}$. 

\end{definition}

For example, for $n=3$, Definition~\ref{def:ptrra} defines a graceful sequence for $\rho$ such that $\rho(t) = \{Maj(t), Maj(t), Min(t)\},$ where we use $Maj: \{0,1\}^n \rightarrow \{0,1\}, Min: \{0,1\}^n \rightarrow \{0,1\}$ to denote the majority (respectively, minority) boolean functions. 

Clearly, implementing a graceful sequence requires exponential space (a counter for every possible decision-type), and is not efficient in any way: We use it to reason about online algorithms for decision-making with small number of agents $n$.  
We note that graceful sequences are a vast class of algorithms (the amount of \textit{periodic} sequences  is  double exponential in $n$), but it is simple in the sense that the decisions for each type are done independent of the other types, or the overall instance. We show that this class does not suffice to guarantee existence of a full $MMS^{adapt}$ with $n=4$. 

\begin{restatable}{theorem}
{pickingSeqFour}
\label{thm:pickingSequence4}
With $n = 4$, any graceful sequence $A_{\rho}$ guarantees at most $\frac{4}{5}-MMS^{adapt}$. 
\end{restatable}

Contrary to this negative result, we can show for the case of $n=4$ that $MMS^{adapt}$ can be fully guaranteed \textit{offline}, and present an algorithm to do so. The algorithm uses a graceful sequence as its main primitive, but leaves a few decisions out that are decided based on the overall situation. 

We start by defining \textit{ambiguous} decision-types as these without a clear majority, i.e., two agents are for the decision and two against it. There are three ambiguous decision-types, which we denote $t_2, t_3, t_4$ by whichever agent agrees with agent $1$. We denote the non-ambiguous, non-consensus decision types $\alpha_1, \ldots, \alpha_4$ by the agent that opposes all others. For ease of notation, we interchangeably use a decision-type $t_i$ or $\alpha_i$ and the amount of decisions of this type in $M$. We let $C$ be the number of consensus decision in $M$. 

Consider the graceful sequence $\rho^*$ that decides the consensus decision-type with the consensus decision, and 
any other non-ambiguous decision type $t$ with:
$$\rho^*(t)_i = \begin{cases} Maj(t) & i \mod 4 \neq 3 \\
Min(t) & i \mod 4 = 3\end{cases}$$ (i.e., follows the majority for the first three occurrences of the decision, and the minority for the fourth one, and so on cyclically).  For any ambiguous decision type $t$, it has:
$\rho^*(t)_i \mod 2 \neq \rho^*(t)_{i+1} \mod 2$ (i.e., makes alternating decisions). 
For each agent $2 \leq i \leq 4$, we let $\eta_i(M) = \sum_{j \neq i} \frac{3}{4}\alpha_j  + \frac{1}{4}\alpha_i + C +  \sum_{j=2}^4 \frac{1}{2}t_j(M)$. 

\begin{algorithm}
\SetAlgoLined
\DontPrintSemicolon
\KwIn{Agents $[n=4]$, decisions $[m]$, preference matrix $M$}
\KwOut{Decision vector $A$}

Let $S = \emptyset, i^* = 1$. 

\For{$i=2$ to $4$} {
    If the ambiguous decision type $t_i$ has an odd number of occurrences, remove the last occurrence of it from $M$ and add it to $S$.
}

Run $A_{\rho^*}$ over $M$.

\If{$\exists i, u_i(A_{\rho^*}) < \eta_i(M) $} {

Let $i^* = i$. 
}

Let agent $i^*$ decide the remaining decisions in $S$. 

 \caption{Deferred Ambiguity Algorithm for $n=4$}
 \label{alg:defered_ambiguity}
\end{algorithm}

\begin{restatable}{theorem}
{deferedAmbig}
\label{thm:deferedAmbiguity}
With $n=4$, for any agent $i$, Algorithm~\ref{alg:defered_ambiguity} guarantees $MMS^{adapt}$. 
\end{restatable}

\subsection{Large Values of $n$}

Finally, for $n\geq 7$, we can show a strong result.

\begin{theorem}
\label{thm:online_mms_impossibility}
No \textit{online} algorithm can guarantee  $MMS^{adapt}$ with $n\geq 7$. 
\end{theorem}



We start by giving a high-level overview of the proof construction. Our goal is to force the algorithm to decide ``too many'' decisions with the minority. In stage I1, we present the algorithm with $n$ decisions with agent $1$ in the minority. The agent should thus be guaranteed at least one decision towards $MMS^{adapt}$ (by separating each decision into its own bundle). For all other decisions, that are decided for the majority, we repeat them to a total of $n$ repetitions in stage I2. We show that the algorithm must be ``graceful'' towards these repeated decisions: If the majority is a $k$ out of $n$ majority, it must decide $k$ of its repetitions for the majority, and $n-k$ for the minority. All in all, for all the repeated decisions, the algorithm is in exact accordance with their value towards $MMS^{adapt}$, and decides one additional decision for the minority. We are able to repeat this trick for stage II, and end up with more repeated graceful decisions, and another decision in favor of the minority (of some agent $\mu \neq 1$ and possibly some other agent). Finally, in stage III, these two decisions for the minority are enough to have a cascading effect that forces the algorithm to only make decisions in favor of the minority, eventually leading to a discrepancy from the $MMS^{adapt}$ value. 

\begin{proof}

\textit{Stage I1: Forcing the algorithm to decide with agent $1$ for the minority.}
\vspace{1mm}

Consider the following $n$ consecutive decisions:

$$S_1 = \begin{bmatrix}
 0 & 0 & \ldots & 0 & 0\\
 0 & 1 & \ldots & 1 & 1 \\
 1 & 0 & \ldots & 1 & 1\\
\ldots \\
 1 & 1 & \ldots & 0 & 1\\
\end{bmatrix}.$$

We have $MMS^{adapt}_1 \geq 1$ by considering $n$ singleton bundles. Thus, there must be some step $t$ where $A_t = 0$ and the algorithm decides for the minority. 

\vspace{1mm}
\textit{Stage I2: Balancing out the majority decisions.}
\vspace{1mm}

After step $t$, instead of proceeding with the rest of the $n-t$ decisions of stage I1, we introduce $n-1$ copies of each of the decision-types preceding step $t$. I.e., if $t=1$, we skip this stage. If $t=2$, we repeat $\begin{bmatrix} 0 \\
0 \\
1 \\
\ldots \\
1\end{bmatrix}$ $(n-1)$ times, etc. 
\begin{claim}
For each of the $t-1$ types that appear $n$ times, the algorithm decides $n-2$ times with the majority, and $2$ times with the minority. 
\end{claim}
\begin{proof}

For agent $1$, it is in the minority of all $t-1$ decisions-types preceding step $t$. Thus, by including exactly one decision of each type in a bundle, (and assigning the additional decision of step $t$ in some arbitrary bundle), the agent is guaranteed at least $MMS^{adapt}_1 \geq 2\cdot (t-1)$. 

Any agent $2 \leq i \leq t$ is in the majority in $t-2$ of the decision-types preceding step $t$, and in the minority of one such decision-type. Thus, similarly, $MMS^{adapt}_i \geq (n-2) \cdot (t-2) + 2$. 

If we let $k_1, \ldots, k_{t-1}$ be the number of times the algorithm decides with the majority for each of the $t-1$ decision-types, and since the algorithm decides with the minority in the decision-type of step $t$, we get that for the algorithm to satisfy $MMS^{adapt}$, the following system of inequalities must hold:
\begin{equation}
\label{eq:system_stage1}
\begin{alignedat}{4}
 & n - k_1 &&+  n - k_2 + \ldots && + n - k_{t-1} &&\geq 2t - 3 \\
 & n - k_1  &&+  k_2 \qquad +  \ldots &&+  k_{t-1}  &&\geq (n-2)\cdot (t-2) + 2 \\
& k_1 &&+  n - k_2 + \ldots &&+  k_{t-1} &&\geq (n-2)\cdot (t-2) + 2 \\
& \ldots \\
 & k_1 &&+  k_2 \qquad + \ldots && +  n - k_{t-1} && \geq (n-2)\cdot (t-2) + 2.
\end{alignedat}
\end{equation}

If we sum up all but the first inequality:
\begin{equation}k_1 + \ldots + k_{t-1} \geq (n-2)(t-1).\end{equation}

The first inequality of Eq.~\ref{eq:system_stage1} can be written as:

\begin{equation}
\label{eq:sum_upper_bound}
(n-2)(t-1) + 1 \geq k_1 + \ldots + k_{t-1},
\end{equation}

Now, let $\Sigma = k_2 + \ldots + k_{t-1}$. Then, Eq.~\ref{eq:sum_upper_bound} can be rewritten as 
$$\Sigma + k_1 \leq  (n-2)(t-1) + 1,$$
and the second inequality of Eq.~\ref{eq:system_stage1} as:
\[
\begin{split}
    & \Sigma - k_1 \geq (n-2)\cdot (t-3), \text{ i.e.,} \\
& 2k_1 \leq (n-2)(t-1) + 1 - (n-2)\cdot (t-3) = 2(n-2) + 1. \end{split}
\]
Hence, $k_1 \leq n - 1.5$. Since $k_1$ is an integer, it must hold that $k_1 \leq n - 2$. We can repeat a similar argument for any $k_i$ up to $k_{t-1}$. By adding together Eq.~\ref{eq:sum_upper_bound} and the second inequality of Eq.~\ref{eq:system_stage1}, we also get that:

\[
\begin{split}
    & k_2 + \ldots k_{t-1} \geq (n-2)(t-2), \text{ and thus,} \\
& k_2 = \ldots = k_{t-1} = n-2,
\end{split}
\]
and $k_1 = n-2$ as well by the third equation of Eq.~\ref{eq:system_stage1}. 
\end{proof}

\vspace{1mm}
\textit{Stage II1: Forcing another minority decision.}
\vspace{1mm}

Let us consider that after stage I1 and I2, we have one decision-type which is not repeated $n$ times, where the algorithm goes by the minority agent $1$ and, possibly, some other agent $\ell$. Let $\mu \not \in \{1, \ell\}$ be some other agent in $[n]$. For simplicity of notation, when we specify our construction we assume $\ell = 2$ and $\mu = 3$, though it holds generally. 

Now consider the following $n-1$ consecutive decisions:

$$S_2 = \begin{bmatrix}
 1 & 1 & \ldots & 1 & 1 & 0\\
 1 & 1 & \ldots & 1 & 1 & 1 \\
 0 & 0 & \ldots & 0 & 0 & 0\\
 0 & 1 & \ldots & 1 & 1 & 1\\
 1 & 0 & \ldots & 1 & 1 & 1\\
 1 & 1 & \ldots & 1 & 1 & 1\\
\ldots \\
 1 & 1 & \ldots & 0 & 1 & 1\\
\end{bmatrix}.$$

Notice that the first $n-3$ decision-types has in the minority agent $\mu$ and another agent that was not in the minority at step $t$. The second to last decision-type has in the minority agent $\mu$ and agent $1$, and the last decision-type has only agent $\mu$ in the minority.

After stage I1 and I2, $\mu$ received exactly its $MMS^{adapt}$ from the algorithm for the repeated decision-type, and the additional minority decision went against it. Together with the additional $n-1$ decisions of this stage, we have \begin{equation}
\label{eq:mu_util}
MMS^{adapt}_{\mu} \geq 1 + u_{\mu}(ALG ; t + (n-1)(t-1)),\end{equation} where $ALG$ is the outcome of the algorithm and $u_{x}(A ; s)$ is the utility of agent $x$ from outcome $A$ by the end of step $s$. We get Eq.~\ref{eq:mu_util} by considering adding one of the $n$ decisions to every bundle. Thus, there must be some step $\tau$ during this stage (i.e., overall it is step $t + (n-1)(t-1) + \tau$), where $ALG[t + (n-1)(t-1) + \tau] = 0$. 

\vspace{1mm}
\textit{Stage II2: Balancing out the majority decisions.}
\vspace{1mm}

After step $\tau$, similarly to stage I2, we introduce $n-1$ copies of each of the $S_2$ decision-types preceding step $\tau$. 

\begin{claim}
For each of the types added at stage II1 before $\tau$, the algorithm decides with the majority exactly according to the number of agents in the majority (i.e., if there are $n-2$ agents in the majority, it decides with it $n-2$ times, and the same for if there are $n-1$ agents in the majority). 
\end{claim}

\begin{proof}
If $\tau \leq n -2$, the argument exactly follows the lines of I2 (albeit, starting with agent $3$ instead of agent $1$). We thus consider $\tau = n - 1$. 

For agent $\mu$, it is in the minority of all the decisions-types added in stage II1, preceding step $\tau$. Thus, by including exactly one decision of each type in a bundle, (and assigning the additional decision of step $t$ in some arbitrary bundle), the agent is guaranteed at least $MMS^{adapt}_1 \geq 2\cdot (\tau-1) + u_{\mu}(ALG ; t + (n-1)(t-1))$. 

Any agent $4 \leq i \leq \min n$ is in the majority in $\tau-2$ of the decision-types preceding step $\tau$, and in the minority of one such decision-type. Thus, similarly, $MMS^{adapt}_i \geq (n-2) \cdot (\tau-2) + 2 + u_i(ALG ; t + (n-1)(t-1))$. 

Agent $2$ is in the majority in all $\tau-1$ of the decision-types preceding step $\tau$, and so $MMS^{adapt}_2 \geq (n-2)(\tau - 1) + 1 + u_i(ALG ; t + (n-1)(t-1)) - 1$ (notice that we deduce $1$ since the algorithm decides for it in step $t$). 

If we let $k_1, \ldots, k_{\tau-1}$ be the number of times the algorithm decides with the majority for each of the $\tau-1$ decision-types, and since the algorithm decides with the minority in the decision-type of step $t$, we get that for the algorithm to satisfy $MMS^{adapt}$, the following system of inequalities must hold:
\begin{equation}
\label{eq:system_stage2}
\begin{split}
  n - k_1 + n - k_2 + \ldots + n - k_{\tau-1} & \geq 2\tau - 3 \\
  n - k_1 + k_2 + \ldots + k_{\tau-1} & \geq (n-2)\cdot (\tau-2) + 3 \\
  k_1 + n - k_2 + \ldots + k_{\tau-1} & \geq (n-2)\cdot (\tau-2) + 3 \\
& \ldots \\
  k_1 + k_2 + \ldots + n - k_{\tau - 2} + k_{\tau-1} & \geq (n-2)\cdot (\tau-2) + 3 \\
  k_1 + k_2 + \ldots + k_{\tau-1} & \geq (n-2)\cdot (\tau-1).
\end{split}
\end{equation}

Similarly to stage I2, we get that $k_i \leq n-2$, for any $1 \leq i \leq \tau - 2$. This is not necessarily true for $k_{\tau - 1}$). 

If we sum all the inequalities except the first one, and apply some arithmetics, we get:

\[
\begin{split}
& k_1 + \ldots + k_{\tau - 2} \geq (n-2)(\tau - 1) - \frac{2k_{\tau - 1} - n + 3}{n - 4} \geq \\
& (n-2)(\tau - 1) - \frac{n + 3}{n - 4} \stackrel{n\geq 7}{\geq} (n-2)(\tau - 2),
\end{split}
\]

and so $k_1 = \ldots = k_{\tau - 2} = n - 2$. 

Then, by the first inequality of Eq.~\ref{eq:system_stage2}, 
$n - 1 \geq k_{\tau - 1}$, and by the second to last inequality of Eq.~\ref{eq:system_stage2}, 
$k_{\tau-1} \geq n - 1$, i.e., $k_{\tau-1} = n-1$. 
\end{proof}

\textit{Stage III: Minority Decisions Cascade}
\vspace{1mm}

We now consider the outcome of the first two stages. We have $(t-1) + (\tau-1)$ decision-types that are repeated $n$ times, for which the algorithm outcome and the $MMS^{adapt}$ guarantee are fully aligned. We also have two decisions decided for the minority, one together with agent $1$ and, possibly, another agent $\ell$, and the second with agent $\mu$ and, possibly another agent, denote it $\mu'$. If $\{1, \ell\} \cap \{\mu, \mu'\} = \emptyset$, then consider the perspective of some other agent $i$: It is in the majority in both decisions, so it agrees with $\mu, \mu'$ on the first, and with $1, \ell$ on the second. Therefore, if the agent includes both decisions in the same bundle for $MMS^{adapt}$, they should get at least one of them decided in their favor, no matter the adversary assignment, and yet the algorithm decided both against agent $i$, so they get less than their $MMS^{adapt}$. 

The interesting case is thus when $\{1, \ell\} \cap \{\mu, \mu'\} \neq \emptyset$, and by our choice of $\mu$ and construction of stage II this can only happen if $\mu' = 1$. For simplicity of notation, we assume for the rest of the proof that $\ell = 2, \mu = 3$. 

Consider the following $n-3$ consecutive decisions:

$$S_{3a} = \begin{bmatrix}
 0 & 0 & \ldots & 0 \\
 1 & 1 & \ldots & 1  \\
 1 & 1 & \ldots & 1 \\
 0 & 1 & \ldots & 1 \\
 1 & 0 & \ldots & 1 \\
 1 & 1 & \ldots & 1 \\
\ldots \\
 1 & 1 & \ldots & 0 \\
\end{bmatrix},$$

i.e., the sequence of decisions where agent $1$ is in the minority together with any other agent other than $\ell$ and $\mu$. We argue that all these decisions must be decided with the minority:
Let decision $1$ and $2$ be the minority decision-types that the algorithm has decided on in steps $t$ and $\tau$ respectively, and let decisions $3$ up to $n-1$ be the decisions we now add. Consider some decision $3 \leq i \leq n-1$ in this sequence. In all decisions $i' < i$, agent $i+1$ (who is in the minority, together with agent $1$, in decision $i$), is in the majority. It therefore has at least $1$ decision in agreement with any other agent in the set of decisions $\{1, \ldots, i\}$. I.e., if the agent includes them in the same bundle for $MMS^{adapt}$, they should expect at least one decision in their favor. Since all decisions $\{1, \ldots, i-1\}$ are decided in favor of the minority (whereas they are in the majority), the algorithm must decide decision $i$ with them for the minority. 

Now consider the following $n-1$ decisions, repeated cyclically $n-1$ times:

$$S_{3b} = \begin{bmatrix}
 0 & 0 & 0 & 0 & \ldots & 0 \\
 0 & 1 & 1 & 1 & \ldots & 1  \\
 1 & 0 & 1 & 1 & \ldots & 1 \\
 1 & 1 & 0 & 1 & \ldots & 1 \\
\ldots \\
 1 & 1 & 1 & 1 & \ldots & 0 \\
\end{bmatrix}.$$

We argue that at cycle $1 \leq c \leq n-1$, decision $1 \leq i \leq n-1$ must be decided for the minority, for reasons very similar to our previous argument:
Agent $i+1$, who is in the minority in decision $i$, is in agreement with every other agent at least $c + 1$ times over the set of all decisions up to and including decision $i$ in cycle $c$. However, for the decisions up to and \textit{not} including decision $i$ in cycle $c$, the algorithm only decided in their favor $c$ times. 

But, this means that by the end of all cycles, each agent (other than agent $1$) gets $n$ decisions in their favor over the decision in steps $t,\tau$ and stage III. However, notice that these are $n$ repetitions of $n-1$ decision-types, with the agent being in a majority of $n-2$ in $n-2$ of decision-types, and in a minority of $2$ in one decision-type. If the agent includes each repetition of a decision-type in a separate bundle for the purpose of $MMS^{adapt}$, they can thus expect a guarantee of $(n-2)\cdot (n-2) + 2 = n^2 - 4n + 6$. This is always strictly greater than $n$ with $n\geq 7$.  

\end{proof}

\section{Maximum Nash Welfare Guarantees}
\label{sec:mnw}

MNW does not always guarantee $MMS^{adapt}$, and in fact it can only guarantee arbitrarily low approximation of it. Before proving this, we see an example that gives intuition into the difference between the MNW outcome and the $MMS^{adapt}$ value. 

\begin{example}
Consider $n=3, m = 15$. In all decisions, there is a single agent opposing the majority: Agent $1$ is the odd one out on $9$ decisions, agent $2$ on $3$ decisions and agent $3$ on $3$ decisions. We have $MMS^{adapt}_1 = 7, MMS^{adapt}_2 = 9, MMS^{adapt}_3 = 9$. These are all attained on a partition where each agent takes $\frac{1}{3}$ of the decisions, uniformly over the different types. Then, an agent can expect $\frac{1}{3}$ of the decisions where it is in opposition to pass, and $\frac{2}{3}$ of the decisions where it is in majority to pass. 

The MNW outcome in this case is the same as the majoritarian outcome. We then have $u_1(A^{MNW}) = 6, u_2(A^{MNW}) = u_3(A^{MNW}) = 12$. Thus, the MNW is only $\frac{6}{7}$-MMS$^{adapt}$. To achieve a MMS$^{adapt}$ outcome, agent $1$ needs to be compensated with at least one decision where it is not in the majority. I.e., MMS$^{adapt}$ is more graceful in this example. 
\end{example}

\begin{theorem}
\label{lem:mnw_mms_adapt}
    There is an instance where MNW can guarantee at most $\frac{6}{n-1}-MMS^{adapt}$ asymptotically. 
\end{theorem}

\begin{proof}
(Proof sketch)

We first give intuition regarding the construction. To extend the example, we wish to have an agent that loses the vast majority of decisions in terms of the majority vote, but would get a sizeable fraction of all decisions in expectation under a random dictator. We then fine-tune the parameters of the construction so that the MNW outcome is the same as the majoritarian one. 

Consider $n+1$ agents, so that $n \mod 3 = 0$ and $n \mod 2 = 1$, and choose some $k$ so that $k \mod n = 0$. Each of the agents $\{2,\ldots, n+1\}$ has $n+1$ decisions where it is the sole agent opposing all others. Agent $1$ has $k(n+1)$ decisions where it is in the minority, divided into three equal parts: $\frac{k}{3}(n+1)$ where it has the same opinion as agents $\{2,\ldots, \frac{n}{3} + 1\}$, and similarly with agents $\{\frac{n}{3} + 2, \frac{2n}{3} + 1\}$ and $\{\frac{2n}{3} + 2, n + 1\}$. We have 
$MMS^{adapt}_1 = k + n(n + \frac{k}{3}).$

The majority vote outcome $A^{maj}$ has $v_1(A^{maj}) = n(n+1), \forall 2 \leq i \leq n+1, v_i(A^{maj}) = (n-1 + \frac{2k}{3})(n+1)$. We notice that the MNW outcome must be of the form where agent $1$ gets $u_1(A^{maj}) + 3t$ of the decisions for some $t\geq 0$, and the other agents get $u_i(A^{maj}) - t$ (i.e., we start from the majority vote outcome, with possibly some of the majority results that go against agent $1$ are reversed). We aim to find a condition on $k$ that guarantees that the maximum is attained at $t = 0$, and hence $A^{MNW} = A^{Maj}$. This happens whenever the maximum over the Nash Welfare when parameterized by $t$ is at some $t < 0$. We write
$$\frac{\partial NW(t)}{\partial t} = \frac{\partial (v_1(A^{maj}) + 3t) \cdot (v_i(A^{maj}) - t)^n }{\partial t} = 0.$$
Solving for $t$, this holds when $t = \frac{1}{3} \left(2 k-n^2+3 n-3\right)$. Thus, to have $t < 0$ it must hold that $k < \frac{1}{2} \left(n^2-3 n+3\right)$. In particular, if we choose $k = \frac{n(n-3)}{2}$, it satisfies this condition. Moreover, it is always an integer number (since $n-3$ is even), and also $n$ is a divisor of $k$ as required. 

With this choice of $k$, we have 
$$ \frac{v_1(A^{MNW})}{MMS^{adapt}_1} = \frac{n(n+1)}{k + n(n + \frac{k}{3})} = \frac{n(n+1)}{\frac{n(n-3)}{2} + n(n + \frac{n(n-3)}{6})},$$

which asymptotically behaves as $\frac{6}{n}$ (notice that the theorem states $\frac{6}{n-1}$, as our construction uses $n+1$ agents).

\end{proof}


\section{Discussion}
\label{sec:discussion}

In our work, we present and analyze a maxi-min-share (MMS) notion for fair decision-making. We believe it captures a notion of \textit{graceful majority}, and it intricately depends on the details of agents' agreements and disagreements over the decisions. In the appendix, we present another possible  \textit{egalitarian MMS} notion that is independent of other agents' preferences, and, since the decisions are binary, expects roughly half of the decisions to be in their favor. 

The main open question is better understanding the existence of $MMS^{adapt}$ outcomes for $n\geq 5$. Our impossibility result has two qualifiers: It holds for \textit{online} algorithms, and it forbids \textit{full} $MMS^{adapt}$. As we see in Section~\ref{sec:mnw}, Maximum Nash Welfare fails to achieve a good approximation of $MMS^{adapt}$, but it is interesting to consider whether any online/offline method is able to do that. 


\bibliography{ref.bib}

\appendix

\section{Reproducibility Checklist}

This paper:

\begin{itemize}
\item Includes a conceptual outline and/or pseudocode description of AI methods introduced (YES)

\item 
Clearly delineates statements that are opinions, hypothesis, and speculation from objective facts and results (YES)

\item 
Provides well marked pedagogical references for less-familiare readers to gain background necessary to replicate the paper (YES)

\item 
Does this paper make theoretical contributions? (YES)
\end{itemize}

If yes, please complete the list below.

\begin{itemize}
\item All assumptions and restrictions are stated clearly and formally. (YES)
\item All novel claims are stated formally (e.g., in theorem statements). (YES)

\item Proofs of all novel claims are included. (YES)

\item Proof sketches or intuitions are given for complex and/or novel results. (YES)

\item Appropriate citations to theoretical tools used are given. (YES)

\item All theoretical claims are demonstrated empirically to hold. (NA)

\item All experimental code used to eliminate or disprove claims is included. (NA)

\item Does this paper rely on one or more datasets? (NO)
\end{itemize}

\clearpage

\section{Missing Proofs for Section~\ref{sec:existence}}

\pickingSeqFour*

\begin{proof}

With $n=4$, there are a total of $8$ decision types (up to equivalence). Out of these, three are ``ambiguous'', in the sense that there is no majority opinion: Namely, 

$$t_2 = \begin{bmatrix} 0 \\ 0 \\ 1 \\ 1\end{bmatrix}, t_3 = \begin{bmatrix} 0 \\ 1 \\ 0 \\ 1\end{bmatrix}, t_4 = \begin{bmatrix} 0 \\ 1 \\ 1 \\ 0\end{bmatrix}.$$

(The index for $t$ represents which agent is in agreement with agent $1$. This will be helpful later in the proof.)

\begin{claim}
For any graceful sequence for $n=4$, there is either an agent $i$ that wins the first occurrence of all three ambiguous decision types, or an agent $j$ that loses the first occurrence of all three.
\end{claim}

\begin{proof}
If agent $1$ wins the first occurrence of exactly one ambiguous decision type $t_i$, then agent $i$ wins all first occurrences of the three ambiguous decision types. 

If agent $1$ wins the first occurrence of exactly two decision types, e.g., $t_2, t_3$, then agent $4$ loses all first occurrences of the three ambiguous decision types (and similarly, for the  symmetric cases, if it wins  $t_2, t_4$, agent $3$ loses all, and if it wins $t_3, t_4$, agent $2$ loses all). 

Other than the above, we have that agent $1$ either wins or loses all first occurrences of the ambiguous decision types. 
\end{proof}

Consider the instance where there is exactly one decision of each ambiguous decision type. If the graceful sequence is such that an agent $j$ loses all first occurrences of an ambiguous type, then $u_j(A_{\rho}) = 0$. However, the agent can partition the decisions to three empty subsets and one subset that contains all three decisions. Then, notice that each other agent agrees with the agent at least once. Therefore, $MMS^{adapt}_j \geq 1$. 

We thus assume going forward that there is an agent $i$ that wins the first occurrence of all three ambiguous types. W.l.o.g., we let it be agent $1$. We thus have $\rho(t_2)_0 = \rho(t_3)_0 = \rho(t_4)_0 = 0.$ 

Let $\alpha_i$ be the decision type such that $\forall j\neq i, \alpha_i^j = 1$, and $\alpha_i^i = 0$. I.e., agent $i$ is opposing all the other agents on the decision. We next show an argument for agent $2$, but it can be adapted for agents $3$ and $4$ as well. 
Consider an instance that have two decisions of type $t_2$, one of type $t_3$, one of type $t_4$, and three of type $\alpha_2 = \begin{bmatrix} 1 \\ 0 \\ 1 \\ 1 \end{bmatrix}$. Then, if $\rho(\alpha_2)_0 = \rho(\alpha_2)_1 = \rho(\alpha_2)_2 = 1$, and $\rho(t_2)_1 = 1$, then $u_2(A_{\rho}) = 1$. However, the partition 

$$\{ \{\alpha_2\}, \{\alpha_2\}, \{\alpha_2\}, \{t_2, t_2, t_3, t_4\},$$
guarantees $MMS^{adapt})_2 \geq 2$. This is because if the adversary assigns the last subset to agent $2$, agent $2$ gets at least $4$ decisions, but if not, then any agent that is assigned the last subset agrees with agent $2$ on at least one decision, and agent $2$ gets to decide some other decision of the first three subsets. This means that for this instance, $u_2(A_{\rho}) \leq \frac{1}{2} MMS^{adapt}_2$. 

Let us now consider if $\rho(t_2)_1 = 0$. Examine the instance that has two decisions of type $t_2$, one decision of type $t_3$, and one decision of type $\alpha_4 = \begin{bmatrix} 1 \\ 1 \\ 1 \\ 0 \end{bmatrix}$. Then, if $\rho(\alpha_4) = 1$, agent $4$ has $u_4(A_{\rho}) = 0$, while $MMS^{adapt}_4 \geq 1$. But if $\rho(\alpha_4) = 0$, consider the instance with one decision of type $\alpha_4$, one decision of type $t_3$, and one decision of type $t_4$. Agent $2$ has $u_2(A_{\rho}) = 0$, but if it partitions to three empty subsets and a subset with all decisions, since it agrees with all agents on at least one decision, it has $MMS^{adapt}_2 \geq 1$. 

If $\rho(\alpha_2)_0 = 0$, then consider the instance with one decision of type $\alpha_2$, one decision of type $t_1$, and one decision of type $t_2$. Agent $4$ has $u_4(A_{\rho}) = 0$, but it agrees with each agent on at least one decision and so has $MMS^{adapt}_4 \geq 1$. 

We thus conclude that for agent $2$, and in fact for any agent $i\in \{2,3,4\}$, the following holds: 
\begin{itemize}
\item $\rho(\alpha_i)_0 = 1$,

\item At least one of $\rho(\alpha_i)_1, \rho(_alpha_i)_2$ is $0$,

\item $\rho(t_i)_0 = 0$.

\end{itemize}

Now consider the instance that has three decisions of type $\alpha_3$, three decisions of type $\alpha_4$, and one decision of each type $t_3, t_4$ respectively. For this instance, $u_2(A_{\rho}) \leq 4$. Consider the partition 

$$\{ \{ \alpha_4, t_3\}, \{\alpha_3, t_4\}, \{\alpha_4, \alpha_3\}, \{\alpha_4, \alpha_3 \} \}. $$

For each bundle of the partition, no matter which agent it is assigned to, agent $2$ gets at least one decision decided in their favor. Moreover, all bundles are of size $2$ and one of them is assigned by the adversary to agent $2$, and thus both decisions are decided in the agent's favor. We conclude that $MMS^{adapt}_2 \geq 5$, and overall that no graceful sequence guarantees more than $\frac{4}{5}-MMS^{adapt}$. 
\end{proof}

\deferedAmbig*

\begin{proof}

Let $\eta_i(M) = \sum_{j \neq i} \frac{3}{4}\alpha_j  + \frac{1}{4}\alpha_i + C +  \sum_{j=2}^4 \frac{1}{2}t_j(M)$, where $C$ is the number of consensus decisions where all agents agree, be the expected value for an agent given that the adversary (when calculating $MMS^{adapt}$ chooses a permutation uniformly at random, as discussed in Claim~\ref{clm:expectation_bound}. 

\begin{claim}
\label{clm:partial_run}
After running $A_{\rho^*}$ over $M$, there is at most one agent $i$ with $u_i(A_{\rho^*}) < \eta_i(M)$, and for this agent $u_i(A_{\rho^*}) \geq \lfloor \eta_i(M) \rfloor \geq \eta_i(M) - 1$. 
\end{claim}

\begin{proof}
We have 
$$u_i(A_{\rho^*}) = \sum_{j \neq i} \lceil \frac{3}{4}\alpha_j \rceil + \lfloor \frac{1}{4}\alpha_i \rfloor + C +  \sum_{j=2}^4 \frac{1}{2}t_j(M),$$
where the summands in the last expression are without $\lfloor \rfloor$ or $\lceil \rceil$ because we made sure that they are exactly integer. For all agents $i$, 
given that there is exactly one floor operation in the expression for $u_i(A_{\rho^*})$. 

we have $u_i(A_{\rho^*}) \geq \lfloor \eta_i(M) \rfloor \geq \eta_i(M) - 1$. 

Now, for two agents $i_1, i_2$ to have $\max \{u_{i_1}(A_{\rho^*}), u_{i_2}(A_{\rho^*}) \} < \eta_i(M)$, it must hold that 

$$\lceil \frac{3}{4} \alpha_{i_1} \rceil + \lfloor \frac{1}{4} \alpha_{i_2} \rfloor <  \frac{3}{4} \alpha_{i_1}  +  \frac{1}{4} \alpha_{i_2}, $$
and
$$\lceil \frac{3}{4} \alpha_{i_2} \rceil + \lfloor \frac{1}{4} \alpha_{i_1} \rfloor <  \frac{3}{4} \alpha_{i_2}  +  \frac{1}{4} \alpha_{i_1}. $$

From the first strict inequality, 
$\alpha_{i_2} \mod 4 \not \in \{0,1\}$, and also $\alpha_{i_2} \mod 4 \neq \alpha_{i_1} \mod 4$. By the second strict inequality, $\alpha_{i_1} \mod 4 \not \in \{0,1\}$. We are left with having $\alpha_{i_1} \mod 4 = 3$ and $\alpha_{i_2} \mod 4 = 2$ or vice versa. But the first option contradicts the first strict inequality and the latter contradicts the second strict inequality. 

\end{proof}

\begin{claim}
\label{clm:ambig_size}
If $|S| = 3$,
$MMS^{adapt}_i \leq \lfloor \eta_i(M \cup S) \rfloor \leq \lfloor \eta_i(M) + 1.5 \rfloor$.


\end{claim}

\begin{proof}
$MMS^{adapt}_i \leq \lfloor \eta_i(M \cup S) \rfloor$ by Claim~\ref{clm:expectation_bound}. 
We also have $\eta_i(M \cup S) - \eta_i(M) = \frac{1}{2} \cdot 3 = 1.5$. 

\end{proof}

\begin{claim}
If $|S| = 3$:

If $u_i(A_{\rho^*} ; M) \geq \eta_i(M)$, then $u_i(A) \geq u_i(A_{\rho^*} ; M) + 1 \geq \eta_i(M) + 1 \geq MMS^{adapt}_i$. 

If $u_i(A_{\rho^*} ; M) < \eta_i(M)$, then $u_i(A) = u_i(A_{\rho^*} ; M) + 3 \geq \eta_i(M) + 2 \geq MMS^{adapt}_i$.
\end{claim}

\begin{proof}
By Claim~\ref{clm:partial_run}, there is at most one agent $i$ with $\eta_i(M) - 1 \leq u_i(A_{\rho^*} ; M) < \eta_i(M)$, and this agent gets $3$ additional decisions in the final phase. Overall, the agent has at least $\eta_i(M) + 2$ decisions in her favor, which by Claim~\ref{clm:ambig_size} is enough to guarantee $MMS^{adapt}_i$. 

For all the other agents that have $u_i(A_{\rho^*} ; M) \geq \eta_i(M)$, notice that since $u_i(A_{\rho^*} ; M)$ is an integer this in fact implies $u_i(A_{\rho^*} ; M) \geq \lceil \eta_i(M) \rceil$. The agents receive an additional decision in the final stage (the ambiguous decision type occurrence where their preference coincides with the deciding agent). Overall, they have $u_i(A_{\rho^*} ; M) \geq \lceil \eta_i(M) + 1 \rceil \geq \lfloor \eta_i(M) + 1.5 \rfloor$, and this guarantees $MMS^{adapt}_i$ by Claim~\ref{clm:ambig_size}. 

\end{proof}

\begin{claim}
If $|S| = 0$, then $MMS^{adapt}_i \leq u_i(A_{\rho^*} ; M)$.
\end{claim}

\begin{proof}
$MMS^{adapt}_i \leq \lfloor \eta_i(M \cup S) \rfloor = \lfloor \eta_i(M) \rfloor \leq u_i(A_{\rho^*} ; M)$ by Claim~\ref{clm:expectation_bound} and Claim~\ref{clm:partial_run}.
\end{proof}

\begin{claim}
If $|S| = 1$, then $MMS^{adapt}_i \leq u_i(A)$.
\end{claim}

\begin{proof}
For an agent $i$ (if such exists) that has $u_i(A_{\rho^*} ; M) \leq \eta_i(M)$, it decides the final two decisions, and so similar to the case of $|S| = 3$ it satisfies the claim. 

For any other agent $i$, by Claim~\ref{clm:expectation_bound} and Claim~\ref{clm:partial_run} we know that for $M$, 

$ MMS^{adapt}_i \leq \lfloor \eta_i(M) \rfloor$, and $u_i(A_{\rho^*} ; M) \geq \lceil \eta_i(M) \rceil$. If $\eta_i(M)$ is not an integer, and since by adding one decision to an instance  $MMS^{adapt}_i$ may only increase by at most $1$ 
, this is enough to establish the claim. 

Otherwise, $\eta_i(M)$ is an integer. If it happens that $MMS^{adapt}_i \leq eta_i(M) - 1$, then we are done, similar to the previous paragraph. Otherwise, since the worst choice of the adversary is equal to the expected uniform choice of the adversary, this means that all choices of the adversary must yield exactly $\eta_i(M)$. Therefore, restricted to $M$, the adversary has complete freedom in choosing the permutation. Let $P$ be a partition of $M \cup S$, and let $j$ be the bundle that is assigned the single decision in $S$. When determining $MMS^{adapt}_i$, the adversary can then assign $j$ to an agent that opposes $i$ on that decision. For all the decision in $M$, this assignment would yield $\eta_i(M)$. Overall, for any partition we see that the adversary can bound $MMS^{adapt}_i \leq \eta_i(M) \leq \lceil \eta_i(M) \rceil \leq u_i(A_{\rho^*} ; M)$. 
\end{proof}

\begin{claim}
If $|S| = 2$, then $MMS^{adapt}_i \leq u_i(A_{\rho^*} ; M)$.
\end{claim}

\begin{proof}
For an agent $i$ (if such exists) that has $u_i(A_{\rho^*} ; M) \leq \eta_i(M)$, it decides the final two decisions, and so similar to the case of $|S| = 3$ it satisfies the claim. 

For any other agent $i$, by Claim~\ref{clm:expectation_bound} and Claim~\ref{clm:partial_run} we know that for $M$, 

$ MMS^{adapt}_i \leq \lfloor \eta_i(M) \rfloor$, and $u_i(A_{\rho^*} ; M) \geq \lceil \eta_i(M) \rceil$. If $\eta_i(M)$ is not an integer, and for $M$, $MMS^{adapt}_i$ happens to satisfy $MMS^{adapt}_i \leq \lfloor \eta_i(M) \rfloor - 1$, the claim holds since $u_i(A) \geq u_i(A_{\rho^*} ; M)$ is then greater by at least $2$ than $MMS^{adapt}_i$, and $MMS^{adapt}_i$ can not increase by more than $2$ since $|S| \leq 2$. Otherwise, if $MMS^{adapt}_i = \lfloor \eta_i(M) \rfloor$, then for any partition there are either:
\begin{itemize}
\item Exactly $6$ permutations where the adversary achieves $\lfloor \eta_i(M) \rfloor$, and $18$ permutations where the adversary achieves $\lceil \eta_i(M) \rceil$ (and $\eta_i(M)$ fractional part is $\frac{3}{4}$)

\item There are more than $6$ permutations by the adversary that achieve $\lfloor \eta_i(M) \rfloor$.
\end{itemize}

This is since $\eta_i(M)$
 itself can not be attained (not an integer value), there are $24$ permutations in total, no permutation achieves a lower value, and so to achieve an expected value over permutation of $\eta_i(M)$, this is required. 
 
 Now let $P$ be some partition of $M$ and let the agent assign the decisions in $S$ as it wishes to the existing bundles in $P$. 

 First, consider if the first condition above holds. If the agent assigns both decisions to the same bundle, there is some permutation where this bundle goes to the agent that opposes agent $i$ in both decisions in $S$, and so the value for this partition is at most $\lceil \eta_i(M) \rceil$ (since no value is added from the two added decisions). If the agent assigns the two decisions to two different bundles, then similarly there is a permutation that assigns each of the bundles to an agent that opposes $i$ and the value is bounded by $\lceil \eta_i(M) \rceil$.

 Consider the latter condition. If the agent assigns both decisions to the same bundle, then there must be a permutation that achieves $\lfloor \eta_i(M) \rfloor$ over $M$ and does not assign the bundle to agent $i$ (since there are exactly $6$ permutations that assign the bundle to agent $i$). Such a choice of permutation would achieve a value of at most $\lfloor \eta_i(M) \rfloor + 1 \leq \lceil \eta_i(M) \rceil$. If the agent assigns the two decisions to two different bundles, then the total number of permutation that avoids assigning any of the decisions to one of the agents that opposes it, is bounded by $6$: The first decision has two choices (either assign to $i$ or to the agent that agrees with it). Then, if the first decision was assigned to $i$, there is exactly one choice of assignment for the other decision in $S$. If it was not assigned to $i$, then there are two choices (either $i$ or the agent that agrees with it over the second decision in $S$). This overall yields three possible assignments and they are multiplied by the two choices of how to assign the remaining two bundles. We conclude that there must be some permutation where at least one of the decisions goes to an agent that opposes it, and for $M$ achieves $\lfloor \eta_i(M) \rfloor$. 

 Finally, if $\eta_i(M)$ is an integer number and $MMS^{adapt}_i = \eta_i(M)$, similarly to the $|S| = 1$ case, we conclude that the adversary has complete freedom in choosing the permutation and it can choose a permutation that assigns both bundles to an agent that opposes it. If $MMS^{adapt}_i \leq \eta_i(M) - 2$, then any way of adding the decisions in $S$ would not exceed $u_i(A) \geq \eta_i(M)$. If $MMS^{adapt}_i = \eta_i(M) - 1$, then one of the following holds:

 \begin{itemize}
\item There are at least $18$ permutations by the adversary that achieve at most $\eta_i(M)$. 

\item There are more than $6$ permutations by the adversary that achieve $\eta_i(M) - 1$.
\end{itemize}

The analysis then follows similarly to the fractional $\eta_i(M)$ case. 

\end{proof}
\end{proof}

\section{Egalitarian MMS}

\begin{definition}
($MMS^{egal}$)

Let $X_1, X_2$ be a partition of $[m]$ (so that $X_1 \cap X_2 = \emptyset, X_1 \cup X_2 = [m]$. 

\[
\begin{split}
&  MMS^{egal}_i = \\
& \max_{\{X_1, X_2\}} \min_{j\in \{1,2\}} |X_j| = \lfloor \frac{m}{2} \rfloor.
\end{split}
\]

In words, egalitarian MMS is the value of the game between agent $i$ that chooses a partition of the decisions into two bundles and an adversary that decides which of the bundles will go as agent $i$ wishes, and which will go the opposite way. 
\end{definition}

We say that an outcome A is $\alpha$-MMS$^{egal}$ if for any agent $i$, $u_i(A) \geq \alpha \cdot MMS^{egal}_i$.
If an outcome is $\alpha$-MMS$^{egal}$ with $\alpha = 1$, we simply say that it is MMS$^{egal}$. 

We demonstrate through a few examples how egalitarian MMS behaves differently from adaptive MMS, and captures a different notion of fairness. 

\begin{example}
\label{ex:egal_adapt_discrep}
Consider $n = 20$ agents and $m$ decisions so that $m \mod n = 0$. For any $k$, for agents $i \in [n-1]$, $M_i[k] = 1$, and $M_n[k] = 0$. 
Then for any agent $i$, $MMS^{egal}_i = \frac{m}{2},$, and for agents $i\in [n-1]$, $MMS^{adapt}_i = \frac{n-1}{n} \cdot m$, while $MMS^{adapt}_n = \frac{m}{n}$. 
\end{example}

Notice several different properties between the two notions. First, $MMS^{egal}_i$ depends only on agent $i$'s valuation, and not on any other agent's valuation. This property is consistent with the notion of MMS for fair allocation. However, we note that for fair decisions, any share-notion that does not depend on the other agents' valuations and is neutral  must be constant, regardless of the specific instance. 

\begin{definition}
A share-notion $S$ is \textit{neutral} if for any $1\leq k \leq m$ and any $M_i, M_i'$ so that for all $1 \leq j \neq k \leq m, M_i[j] = M_i'[j]$, and $M_i[k] \neq M_i'[k]$,
$S(M_i) = S(M_i')$. I.e., the share is agnostic to whether the agent is for decision $k$ or against it.
\end{definition}

The following lemma is then immediate:
\begin{lemma}
Any share notion $S$ that only depends on $M_i$ and is neutral has for any $M$ $S(M) = C$ for some constant. 
\end{lemma}

As we mention in the introduction, the share notions of \cite{fairPublicDM} in the binary setting, namely Round-Robin-Share, Pessimistic-Proportional-Share, and PROP1, are all equal to $\lfloor \frac{m}{n} \rfloor$, regardless of the preference profile. 


The following lemma is then almost immediate:

\begin{restatable}{lemma}{MMSegalAdapt}
\label{lem:mms_egal_adapt}
    An outcome that guarantees $\alpha$-MMS$^{egal}$ guarantees at least $\frac{1}{2}\alpha$-MMS$^{adapt}$. 
\end{restatable}

For $n=2$, it is obvious that $MMS^{egal}$ can be guaranteed: Just rule roughly half of the decisions in favor of agent $1$, and half in favor of agent $2$. With $3$ agents, it is less obvious that an arrangement can be found that guarantees $MMS^{egal}$ for all agents. With $4$ agents and more, Example~\ref{ex:mms_vs_rds} shows that no approximation of $MMS^{egal}$ can be guaranteed. 

We now present an algorithm that settles the question of $MMS^{egal}$ existence for $n=3$ in the affirmative. 
In the algorithm, each decision is chosen in favor of the  majority, from the set of agents with utility less than $\lfloor \frac{m}{2} \rfloor$, with a "tie-breaker" in favor of the agent with the smallest utility. Once all agents receive their $MMS^{egal}$ guarantee, we return to rule in favor of the overall majority. 

\begin{restatable}{theorem}
{muffledMaj}
\label{thm:muffledMaj}
    For $n = 3$, Muffled Majority results in every agent's utility being at least $\lfloor \frac{m}{2} \rfloor $.
    $\forall 1\leq i \leq n,  u_i(A) \geq \lfloor \frac{m}{2} 
    \rfloor $. 
\end{restatable}

\begin{proof}
    Let $j^*$ be the index $j\in [m]$ so that in the $j^*$ loop of the outer for loop, some agent ends with $s_i = \frac{m}{2}$. For every $1 \leq j \leq j^*$, the decision is by the majority of the three agents, and so there are at least two agents $i$ that get a decision in their favor. We thus have \begin{equation}
    \label{eq:j*}
    s_1 + s_2 + s_3 \geq 2 \cdot j^*.\end{equation}
    Assume w.l.o.g. that $s_1 \geq s_2 \geq s_3$ after the loop run $j^*$. Then, $2 \cdot j^* \leq s_1 + s_2 + s_3 \leq 1.5 (s_1 + s_2)$, and so $j^* \leq \frac{3}{4} (s_1 + s_2) \leq 1.5 s_1 \leq \frac{3}{4}m < m$. Thus, we conclude that every execution of the algorithm arrives at such $j^*$ before the conclusion of the outer for loop. 

    Let $j^{**}$ the lowest index $j \in [m]$ so that in the $j^{**}$ loop of the outer for loop, at least two agents $i$ end with $s_i \geq \lfloor \frac{m}{2} \rfloor $. We wish to show that $j^{**} \leq m$. If $j^{**} \leq j^* < m$, we are done. Otherwise, we know that there is at least one agent $i$ that at $j^{**}$ has $s_i \geq \lfloor \frac{m}{2} \rfloor$, as it has $s_i = \lfloor \frac{m}{2} \rfloor$ at $j^*$ and $s_i$ are monotone increasing in the run of the algorithm. Let $k, \ell$ be the other agents, and assume w.l.o.g. $s_k \geq s_{\ell}$, then since $j^{**}$ is the minimal index that has two agents with $s_i \geq \lfloor \frac{m}{2} \rfloor $, we know that $s_k = \lfloor \frac{m}{2} \rfloor $. So,

    \begin{equation}
    \label{eq:j**}
    \begin{split}
    & m \geq s_{\ell} + \lfloor \frac{m}{2} \rfloor = s_k + s_{\ell} \stackrel{(i)}{\geq} 2 \cdot j* - \lfloor \frac{m}{2} \rfloor + j^{**} - j^* = \\
    & j^* + j^{**} - \lfloor \frac{m}{2} \rfloor \geq \lfloor \frac{m}{2} \rfloor + j^{**} - \lfloor \frac{m}{2} \rfloor,
    \end{split}
    \end{equation}
    
    where transition (i) is because Equation~\ref{eq:j*} holds at step $j*$, and so we have $s_k + s_{\ell} \geq 2\cdot j* - \lfloor \frac{m}{2} \rfloor$. Then, during steps $j^* + 1, \ldots, j^{**}$, the two agents gain at least $j** - j*$ to their $s_i$ counters.
    The last transition is since $j* \geq \lfloor \frac{m}{2} \rfloor$, as each agent gets at most one decision in their favor in every run of the loop. We conclude that there is a run of the loop where there are at least two agents with $s_i \geq \frac{m}{2}$. 

Let $j^{***}$ be the lowest index $j \in [m]$ so that in the $j^{***}$ loop of the outer for loop, all three agents $i$ end with $s_i \geq \lfloor \frac{m}{2} \rfloor$. If $j^{***} \leq j^{**}\leq m$, then the algorithm achieves its goal. Otherwise, we know that there is exactly one agent $\ell$ that has $s_{\ell} < \lfloor \frac{m}{2} \rfloor$ by the end of the $j^{**}$-th run of the loop. This agent will be the only one that is not muffled in the steps $j^{**} + 1, \ldots, j^{***}$, and so it gets all these decisions in its favor. We thus have 

\begin{equation}
\label{eq:j***}
\begin{split}
& \lfloor \frac{m}{2} = s_{\ell} \stackrel{(ii)}{\geq} j^* + j^{**} - 2\lfloor \frac{m}{2} \rfloor + j^{***} - j^{**} = \\
& j^* + j^{***} - 2\lfloor \frac{m}{2} \rfloor \geq \lfloor \frac{m}{2} \rfloor + j^{***} - 2\lfloor \frac{m}{2} \rfloor = j^{***} - \lfloor \frac{m}{2} \rfloor,
\end{split}
\end{equation}

where transition $(ii)$ is because Equation~\ref{eq:j**} 
has $s_k + s_{\ell} \geq j^* + j^{**} - \lfloor \frac{m}{2} \rfloor$ at the end of the $j^{**}$-th loop run, and at that point $s_k = \lfloor \frac{m}{2} \rfloor$. Equation~\ref{eq:j***} yields $j^{***} \leq 2 \lfloor \frac{m}{2} \rfloor \leq m$. 
\end{proof}

Theorem~\ref{thm:muffledMaj} together with Lemma~\ref{lem:mms_egal_adapt} establish that the algorithm guarantees at least $\frac{1}{2}-MMS^{adapt}$. However, we can use the algorithm as a primitive to achieve a better guarantee of $\frac{3}{4}-MMS^{adapt}$ (while maintaining $MMS^{egal}$). We first establish a few important technical claims regarding $MMS^{adapt}$.

\begin{restatable}{lemma}
{muffledAdapt}
\label{thm:muffledMajMMSAdapt}
    For $n = 3$, Muffled Majority guarantees $\frac{3}{4}-MMS^{adapt}$ (and this analysis is tight).
\end{restatable}

To show this, we first slightly improve the upper bound of Lemma~\ref{lem:mms_adapt_n=3} for the agent with the \textit{least} $MMS^{adapt}$ value.

\begin{corollary}
\label{corr:mms_adapt_n=3_improved_bound}
For $n=3$, 
$$\min_i MMS_i^{adapt} \leq m - \lceil \frac{4\delta}{9} \rceil,$$
where $\delta$ is the number of non-consensus decisions. 

\end{corollary}

\begin{proof}

\[
\begin{split}
& \sum_{i=1}^3 MMS_i^{adapt} \\
& \leq \lfloor \frac{2}{3}\delta_2 + \frac{2}{3} \delta_3 + \frac{1}{3} \delta_1 + m - \delta \rfloor \\
& + \lfloor \frac{2}{3}\delta_2 + \frac{2}{3} \delta_1 + \frac{1}{3} \delta_3 + m - \delta \rfloor \\
& + \lfloor \frac{2}{3}\delta_1 + \frac{2}{3} \delta_3 + \frac{1}{3} \delta_2 + m - \delta \rfloor \\
& \leq 3m - 3\delta + \lfloor \frac{5}{3} \delta \rfloor \\
& = 3m - \lceil \frac{4}{3} \delta \rceil
\end{split}
\]

We then have
\[
\begin{split}
& \min_i MMS_i^{adapt} \leq \frac{1}{3} \sum_{i=1}^3 MMS_i^{adapt} \leq \\
& m - \frac{1}{3}\lceil \frac{4}{3} \delta \rceil \leq m - \lceil \frac{4}{9} \delta \rceil. 
\end{split}
\]

\end{proof}

We now prove Lemma~\ref{thm:muffledMajMMSAdapt}:

\begin{proof}
(Upper bound) Consider $\forall j, M_1[j] = 0, M_2[j] = M_3[j] = 1$. I.e., agent $1$ opposes agents $2$ and $3$ in all $m$ decisions. A quick corollary of Theorem~\ref{thm:muffledMaj} is then (assuming an even $m$) that all agents have $v_i(A) = \frac{m}{2}$. However, direct calculation shows that  $MMS_2^{adapt} = MMS_3^{adapt} = \frac{2}{3}$, and so the algorithm guarantees $\frac{\frac{1}{2}}{\frac{2}{4}} = \frac{3}{4}-MMS^{adapt}$. 

(Lower bound) We take the perspective of one of the agents, w.l.o.g. agent $1$. We know by Corollary~\ref{corr:mms_adapt_n=3_bound} that $MMS_1^{adapt} \leq m - \lceil \frac{\delta}{3} \rceil. $. We know by Theorem~\ref{thm:muffledMaj} that if we arrange all the decisions of type $\delta$ (non-unanimous decisions) first, and run the algorithm with $m = \delta$, each agent is guaranteed at least $\lfloor \frac{\delta}{2} \rfloor$. Moreover, we can make sure that the last out of $\delta$ non-unanimous decisions is such that the agent that has minimal $MMS_i^{adapt}$ is in minority. 
Then, we decide all unanimous decisions in the obvious way, and overall achieve $v_1(A) \geq \lfloor \frac{\delta}{2} \rfloor + m - \delta = m - \lceil \frac{\delta}{2} \rceil$ for the minimal $MMS_i^{adapt}$ agent, and $m - \lfloor \frac{\delta}{2} \rfloor$ for the other two agents. 

Then, if agent $1$ is a higher $MMS_i^{adapt}$ agent, we have:

\[
\begin{split}
& \frac{u_1(A)}{MMS_1^{adapt}} \geq \frac{m - \lfloor \frac{\delta}{2} \rfloor}{m - \lceil \frac{\delta}{3} \rceil} \geq \frac{m - \frac{\delta}{2}}{m - \frac{\delta}{3}} \geq \\
& \frac{3}{2} \cdot \frac{2m - \delta}{3m - \delta} \stackrel{\delta \leq m}{\geq} \frac{3}{2}\frac{m}{2m} = \frac{3}{4}. 
\end{split}
\]

If agent $1$ is a lower $MMS_i^{adapt}$ agents, and $delta \geq 9$, we have:

\[
\begin{split}
& \frac{u_1(A)}{MMS_1^{adapt}} \geq \frac{m - \lceil \frac{\delta}{2} \rceil}{m - \lceil \frac{4\delta}{9} \rceil} \geq \frac{m - \lceil \frac{\delta}{2} \rceil}{m - \lceil \frac{\delta}{3} - \frac{\delta}{9} \rceil} \geq \\
& \frac{m - \lceil \frac{\delta}{2} \rceil}{m - \lceil \frac{\delta}{3} \rceil - 2} \geq \frac{m - \frac{\delta}{2} - 1}{m - \frac{\delta}{3} - 2} \geq \frac{m - \frac{\delta}{2}}{m - \frac{\delta}{3}} \geq \frac{3}{4}, 
\end{split}
\]
where the last transition is done similarly to the previous case.

But for $\delta < 9$, it holds that $\lceil \frac{\delta}{2} \rceil = \lceil \frac{4\delta}{9} \rceil$, and so the ratio is always $1$ in this case. 
\end{proof}

We know from Example~\ref{ex:egal_adapt_discrep} that $MMS^{egal}$ and $MMS^{adapt}$ are incompatible, in the sense that it is impossible to always find an outcome that guarantees both. We thus now limit our attention to search for an algorithm that guarantees $MMS^{adapt}$. 

\begin{restatable}{lemma}{MNWmmsEgal}
\label{lem:mnw_mms_egal}
    MNW guarantees at least $(\frac{2}{n} - \frac{2}{m})-MMS^{egal}$, and there is an instance where MNW guarantees at most $\frac{2}{n}-MMS^{egal}$. 
\end{restatable}

\begin{proof}
    (Lower bound)

    The lower bound is a direct corollary of Theorem 4 of \cite{fairPublicDM} that guarantees $\lfloor \frac{m}{n} \rfloor$. Then, $$\frac{ \lfloor \frac{m}{n}\rfloor}{\lfloor \frac{m}{2} \rfloor} \geq \frac{\frac{m}{n}-1}{\frac{m}{2}} = \frac{2(m - n)}{n\cdot m} = \frac{2}{n} - \frac{2}{m}. $$

    (Upper bound)

Consider the case in which for all $i \in \{1, n-1\}$ and for all $j \in m$, $M_i[j] = 1$ and $M_n[j] = 0$. Assume also $m \mod n = 0, m \mod 2 = 0$. 

Let $\theta$ be the amount of decisions $j$ that in the MNW have $A^{MNW}[j] = 1$ (i.e., the decision goes against agent $n$). And the rest of the $m-\theta$ decisions have $A^{MNW}[j] = 1$. 

The utility of anyone from the first $n-1$ players is $\theta$.
The utility of the last player n is $m-\theta$.

the NW of this outcome is $\theta^{n-1}(m-\theta) = m\theta^{n-1}-\theta^n$. 

We solve the following equation for the derivative in order to find the value of $\theta$ that results in $MNW$:

\begin{equation}
\label{eq:NW_deriv}
\begin{split}
& \frac{\partial (m\theta^{n-1}-\theta^n)}{\partial \theta}= m(n-1)\theta^{n-2}-n\theta^{n-1}=\\
& \frac{\theta^{n-2}}{n}(m-\frac{m}{n}-\theta)=0
\end{split}
\end{equation}

which yields $\theta = m - \frac{m}{n}$. 

The $MMS^{egal}$ is $\frac{m}{2}$, and the utility of the worst-off player is $\frac{m}{n}$, so any outcome must be at most $\frac{m}{n} = \frac{2}{n} \cdot \frac{m}{2} =\frac{2}{n}-MMS^{egal}$.
    
\end{proof}

\end{document}